\documentclass [12pt]{article}
\usepackage{graphicx,amssymb,amsfonts,latexsym,amsmath,amsthm,times}
\usepackage{epsfig}
\usepackage{fancyhdr}
\usepackage{color}
\usepackage[english]{babel}
\usepackage{enumerate}
\numberwithin{equation}{section}

\usepackage{doi}

\newtheorem{theorem}{Theorem}

\newtheorem{example}[theorem]{Example}

\newtheorem{lemma}[theorem]{Lemma}

\newtheorem{remark}[theorem]{Remark}

\newcommand{\ex}{\mathbf {E}}
\newcommand{\pr}{\mathbf {P}}
\newcommand{\R}{\mathbb R}
\newcommand{\Rd}{\mathbb R^d}

\newcommand{\spctim}{\mathbb R^{d+1}}
\newcommand{\del}{\partial }
\newcommand{\1}{\mathbf 1}
\newcommand{\eps}{\varepsilon}
\newcommand{\bs}{\boldsymbol}

\begin{document}

\footnotesize {\flushleft \mbox{\bf \textit{Math. Model. Nat.
Phenom.}}}
 \\
\mbox{\textit{{\bf Vol. 3, No. 1, 2008, pp. 1-3}}}

\thispagestyle{plain}

\vspace*{2cm} 
\normalsize \centerline{\Large \bf A Semi-Markov Algorithm}
\vspace*{2mm}
\centerline{\Large \bf for Continuous Time Random Walk Limit Distributions}

\vspace*{1cm}

\centerline{\bf G. Gill$^a$, P. Straka$^a$
\footnote{Corresponding
author. E-mail: p.straka@unsw.edu.au}}

\vspace*{0.5cm}

\centerline{$^a$ School of Mathematics \& Statistics, UNSW Australia}




\vspace*{1cm}

\noindent {\bf Abstract.}
The Semi-Markov property of Continuous Time Random Walks (CTRWs) and
their limit processes is utilized, and the probability distributions
of the bivariate Markov process $(X(t),V(t))$ are calculated: 
$X(t)$ is a CTRW limit and $V(t)$ a process tracking the age, 
i.e. the time since the last jump.
For a given CTRW limit process $X(t)$, a sequence of discrete CTRWs 
in discrete time is given which converges to $X(t)$ (weakly
in the Skorokhod topology). Master equations for the discrete
CTRWs are implemented numerically, thus approximating
the distribution of $X(t)$. 
A consequence of the derived algorithm is that any distribution of 
initial age can be assumed as an initial condition for the CTRW limit dynamics. 
Four examples with different temporal scaling are discussed: 
subdiffusion, tempered subdiffusion, the fractal mobile/immobile model 
and the tempered fractal mobile/immobile model.

\vspace*{0.5cm}

\noindent {\bf Key words:} anomalous diffusion,
fractional kinetics,
Semi-Markov,
fractional derivative

\noindent {\bf AMS subject classification:} 60F17, 60G22, 90C40


\vspace*{1cm}

\setcounter{equation}{0}
\section{Introduction}

Subdiffusion is now a well-studied theoretical phenomenon in 
statistical physics, motivated by experimental findings in many
different fields, most prominently biophysics 
\cite{Metzler2000,regner2013anomalous,Berkowitz06,Scalas2006Lecture,
Santamaria2006a,Banks2005}.
The Continuous Time Random Walk (CTRW) has been a particularly
successful model for subdiffusion \cite{Metzler2000,HLS2010b},
due to both its tractability and flexibility:
i) Probability densities can be computed via the
fractional Fokker-Planck equation \cite{HLS10PRL,Langlands2005a};
ii) Reaction-subdiffusion equations can be derived from CTRW dynamics
\cite{Mendez2010,Angstmann2013}
iii) Nonlinear dynamics may be incorporated into CTRWs
\cite{StrakaFedotov14,Fedotov2015};
iv) CTRWs, via subordination, can model a variety of scaling 
behaviours and cross-overs between scales
(see \cite{Stanislavsky2008,SchumerMIM} and Section 6 in this article);
and
v) Via a coupling between jumps and waiting times, an even greater
variety of CTRW processes can be modeled 
\cite{StrakaHenry,Jurlewicz}, with applications to L\'evy Walks
\cite{Magdziarz2015} and relaxation phenomena
\cite{Weron10}.

CTRWs and their scaling limits, however, do not possess the Markov 
property, but are in fact Semi-Markov processes
\cite{SemiMarkovCTRW}.
This means that the calculation of the joint distribution at multiple
times (termed ``finite-dimensional distributions'' in stochastic process
theory) is problematic, though significant progress has been made 
\cite{Busani2015,MMMPS12,Baule2007}.

\begin{figure}
\centering
\includegraphics[width=0.8\textwidth]{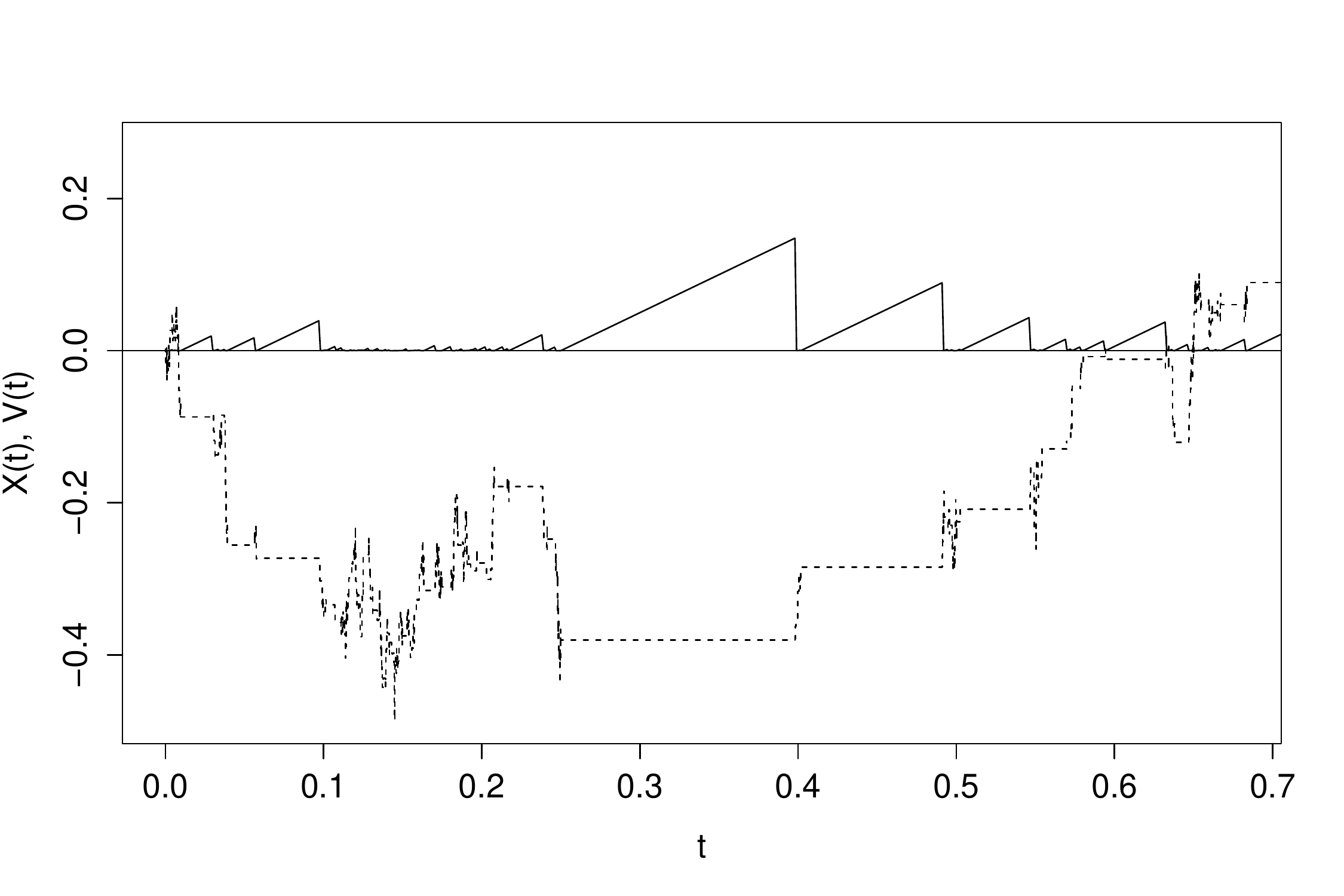}
\caption{\label{fig:V}
Sample paths of the age process $V(t)$ (full line), which renders the CTRW limit
$X(t)$ (dashed line) Markovian.}
\end{figure}

In this article, we utilise the Semi-Markov property of \emph{scaling
limits} of CTRWs, and thus derive a computational algorithm for
the calculation of the probability distributions of CTRW limit
processes. 
Our approach uses the purely Markovian dynamics of $(X(t),V(t))$,
where $X(t)$ is a CTRW limit process and $V(t)$ the process
which tracks the time which has passed since the last jump. 
This process has saw teeth sample paths
(Figure~\ref{fig:V}, also see \cite{MMMPS12})
and is well-known from renewal theory as the ``age'' or 
``backward recurrence time.''
Here, we shall refer to $V(t)$ as the ``residence time.''

The main conceptual difficulty with the Semi-Markov property of 
CTRW limits is that conditional on $V(t) = 0$ we almost surely have 
$V(t') = 0$ for infinitely many $t'$ in
$(t,t+\varepsilon)$ for any $\varepsilon > 0$.
A careful analysis of the limiting sample paths is necessary 
to properly define $V(t)$ and to
establish the Markov property \cite{SemiMarkovCTRW}.
It is seemingly necessary to utilize jump processes with infinite
L\'evy measures to define the joint process $(X(t),V(t))$. 
The procedure that we use to approximate these is similar to 
the approximation of L\'evy processes by compound Poisson processes
(see e.g.\ Section 3.4 in \cite{MeerschaertSikorskii}). 
We walk the reader through the main technical steps in Sections 2-4.

Our algorithm (Section 5) computes the probability densities 
of $(X(t),V(t))$, and by the Markov property and the 
Chapman-Kolmogorov equations, joint distribution of this process at
multiple times $t_1, \ldots, t_k$ can be calculated. 
By taking marginal distributions,
one thus arrives at the joint distribution of $X(t)$ at multiple
times.

Another important application of our algorithm in the fact
that any age distribution may be taken as an initial condition.
This is an important generalization to the Fractional-Fokker-Planck
equation, which implicitly assumes that the initial age of every 
particle equals $0$. 
For instance, taking a snapshot of a cell in which protein molecules are
undergoing (tempered) subdiffusion, there is no reason to believe that 
the time of the snapshot marks the beginning of a waiting time
for each protein molecule. 
We deem it more likely that an ``equilibrium'' initial condition
for the molecule residence times is more appropriate (Section 6).

\section{CTRWs as random walks in space-time}

In this section we set up the theory for scaling limits of CTRWs. 
Space and time need to be jointly rescaled in order to arrive at a meaningful limit,
in much the same fashion as Brownian motion is the scaling limit
of random walks. 
Which scaling functions are appropriate will depend on the tail behaviour of the
waiting time and jump distributions. 
For simplicity, we will later assume nearest neighbour jumps,
and focus on what scaling limits are appropriate for the waiting times,
but the derivation in this Section is held as general as possible, 
which is of independent interest, and causes no extra difficulty.

The key property of CTRWs, which makes much of their analysis a great deal easier compared
to e.g.\ fractional Brownian motion, is the renewal property: Every time a walker jumps, 
its entire future trajectory becomes independent of its past. The next jump time and the 
next position thus only depend on the current time and position; in
other words, position and jump time constitute a Markov chain in space-time $\spctim$. The probability distribution of this Markov
chain is then uniquely determined by i) its starting point in space-time and 
ii) a jump kernel $K(dz,dw|x,s)$ expressing the probability
that conditional on a CTRW arriving at $x$ at time $s$, 
its next jump happens at time $s+w$ and is of size $z$.
It satisfies that 

\begin{enumerate}
\item
$B \times C \mapsto K(B \times C | x,s)$ is a probability measure on $\Rd \times (0,\infty)$ for every $(x,s) \in \spctim$
\item
$(x,s) \mapsto K(B \times C | x,s)$ is measurable for any (Borel) $B \times C \subset \spctim$.
\end{enumerate}

For example, to define a subdiffusive random walk with 
subdiffusive coefficient $0 < \beta < 1$ in a space- 
and time-dependent external force field $b(x,t)$, define the
transition probability kernel via
\begin{align*}
K(B \times (w, \infty) | x,s) = (1 \wedge w^{-\beta}) \mathcal N(B| b(x,s+w), \sigma^2), \quad B \subset \Rd, \quad w > 0,
\end{align*}
where $\wedge$ stands for ``minimum'' and 
$\mathcal N(dz | \mu, \sigma^2)$ denotes a Gau\ss{}ian 
probability distribution on $\mathbb R$ with mean $\mu$ and 
variance $\sigma^2$.
Note that the jump, occurring at time $s+w$, is biased by the
external force $b(x,t)$, which is accordingly evaluated at the time 
$s+w$.

The above Markov chain defines a sequence of random points in space-time
$(x,s) = (A_0, D_0)$, $(A_1, D_1)$, $(A_2, D_2)$, $\ldots$
from which the CTRW trajectory $(X(t))_{t \ge s}$ can be uniquely 
reconstructed: If $D_{k} \le t < D_{k+1}$, then $X(t) = A_k$.
To avoid confusion, we stress that there are two different
notions of ``time'': CTRW jumps occur in \emph{physical time} 
(which we denote by $t$), 
at epochs given by $D_k, n \in \mathbb N$.  
The jumps of the space-time Markov chain 
$(A_k, D_k)_{k \in \mathbb N_0}$ occur at the integer times 
$k \in \mathbb N$, which
corresponds to the count of CTRW jumps.
In the scaling limit below, this count becomes continuous,
and we dub it the \emph{auxiliary time} (usually writing $r$).

We identify a CTRW with its underlying space-time Markov chain.
We then give conditions for a sequence of such Markov chains to
converge to a continuum ``jump-diffusion'' process, whose state 
space is $\spctim$ (Theorem~\ref{thm:spctim-conv}). 
This convergence holds on the stochastic process level, 
in the sense of weak convergence of probability measures on the
Skorokhod space of trajectories. 
Trajectories of this jump-diffusion then again map to
trajectories of CTRW limit processes (Theorem~\ref{th:SM-conv}).

\begin{theorem}
\label{thm:spctim-conv}
For every $n \in \mathbb N$, 
let $(A^n, D^n) = \{(A^n_k, D^n_k)\}_{k \in \mathbb N_0}$
be a Markov chain on the state space $\spctim$ 
with starting point $(x_0,s_0)$ and a transition kernel 
$K^n$ as described above. 
Assume that 
\begin{enumerate}
\item
\begin{align}
\lim_{\varepsilon \downarrow 0}\lim_{n \to \infty}
n \int\limits_{\|z\| < \varepsilon} \int\limits_{ 0 \le w < \eps} z_i K^n(dz,dw|x,s) &= b_i(x,s), ~~~ 1 \le i \le d
\label{eq:K1}\\
\lim_{\eps \downarrow 0} \lim_{n \to \infty}
n \int\limits_{\|z\| < \eps} \int\limits_{0 \le w < \eps} w K^n(dz,dw|x,s) &= c(x,s)
\label{eq:K2}\\
\lim_{\eps \downarrow 0} \lim_{n \to \infty}
n \int\limits_{\|z\| < \eps} \int\limits_{0 \le w < \eps}  z_i z_j K^n(dz,dw|x,s) &= a_{ij}(x,s), ~~~ 1 \le i,j \le d
\label{eq:K3}\\
\begin{split}
\lim_{n \to \infty}
n \int\limits_{z \in \Rd} \int\limits_{0 \le w} g(z,w) K^n(dz,dw|x,s)
&=  \int\limits_{z \in \Rd} \int\limits_{0 \le w} g(z,w) \Pi(dz,dw|x,s)
\label{eq:K4}
\end{split}
\end{align}
where $a_{ij}$, $b_i$ and $c$ are real-valued bounded continuous functions,
$z = (z_1, \ldots, z_d)$, 
$\Pi(\cdot \times \cdot | x,s)$ is a L\'evy measure on $\Rd \times [0,\infty)$ (see remark
below) for every $(x,s) \in \spctim$
and $g$ is varying over all real-valued bounded continuous functions which vanish in a neighborhood of the origin $(0,0)$.
\item
The operator $\mathcal A$ given by
\begin{multline}
\label{eq:DAgen}
\mathcal A f(x,s)
=  b_i(x,s) \del_{x_i} f(x,s)
+ c(x,s) \del_s f(x,s)
+\frac{1}{2}  a_{ij}(x,s) \del_{x_i} \del_{x_j} f(x,s)\\
+\int\limits_{z \in \Rd} \int\limits_{w \ge 0} \left[f(x+z,s+w)
-f(x,s)-  z^i \1(\|z\| < 1) \del_{x_i} f(x,s)\right]  \Pi(dz,dw | x,s)
\end{multline}
generates a Feller semigroup of transition 
probabilities\footnote{$T_r(dy,dt|x,s)$ denotes the probability that
$A(r) \in dy$ and $D(r) \in dt$ given $A(0)=x$, $D(0)=s$.
It thus operates on continuous functions vanishing at $\infty$, via
$T_r f(x,s) = \iint f(y,t) T_r(dy,dt|x,s)$.
The semigroup property reads $T_r T_{r'} f = T_{r + r'}f$,
and is equivalent to the Chapman-Kolmogorov equations for Markov
processes. 
The Feller property is a technical condition, 
see e.g.\ \cite{Applebaum}.} 
$(T_r)_{r \ge 0}$ on $C_0(\spctim)$ (the space of
real-valued continuous functions which vanish at $\infty$). 
\item
$\{L(r)\}_{r \ge 0}$ is an independent Poisson process with unit
intensity.
\end{enumerate}
Then the sequence of processes 
$\left\lbrace\left(A^n_{L(nr)}, D^n_{L(nr)}\right)\right\rbrace_{r \ge 0}$ 
converges weakly (with respect to the Skorokhod $J_1$ topology) to the 
$\spctim$-valued diffusion process with jumps 
$\{(A(r),D(r))\}_{r \ge 0}$ 
starting at $(x_0,s_0)$ and governed by the Feller semigroup 
$(T_r)_{r \ge 0}$.
\end{theorem}

A proof is given in the appendix.

\begin{remark}
A sufficient condition for \eqref{eq:DAgen} to be the generator
of a Feller semigroup is that the coefficients 
$a_{ij}(x,s)$, $b_i(x,s)$, $c(x,s)$ and
$\Pi(\cdot | x,s)$ satisfy certain growth and Lipschitz conditions
\cite[Ch~6]{Applebaum}.
In this case there exist unique solutions to stochastic differential
equations whose semigroup is $(T_r)_{r \ge 0}$. 

\end{remark}

\begin{remark}
That $\Pi(\cdot | x,s)$ is a L\'evy measure for every 
$(x,s) \in \spctim$ means that it is supported on
$\spctim \setminus \{(0,0)\}$ and satisfies
\begin{align*}
\int_{z \in \Rd} \int_{w \ge 0} 
\left(1 \wedge \|(z,w)\|^2 \right) \Pi(dz,dw|x,s) < \infty.
\end{align*}
Since all measures $K^n(\cdot | x,s)$ are supported on 
$\Rd \times (0,\infty)$ (i.e.\ waiting times are strictly
positive) it follows that $\Pi(\cdot | x,s)$ is in fact supported 
on $\Rd \times [0,\infty) \setminus \{(0,0)\}$.
Readers familiar with L\'evy processes will recognize that
the requirement that the limiting process $D(t)$ be strictly
increasing a.s.\ in fact is equivalent to
$$\int_{z \in \Rd} \int_{w \ge 0} 
\left(1 \wedge (\|z\|^2 + w) \right) \Pi(dz,dw|x,s) < \infty.$$
\end{remark}

\begin{example}
\label{ex:AD}
Define the kernels $K^n$ via 
\begin{align*}
K^n(B \times (w, \infty) | x,s) 
= \left(1 \wedge \frac{w^{-\beta}}{n \Gamma(1-\beta)} \right)
 \mathcal N(B| \bs b(x,s+w) / n, \sigma^2/n \cdot \bs I), 
\quad B \subset \Rd, \quad w > 0,
\end{align*}
where $\Gamma$ is the Gamma-function, $0 < \beta < 1$,
$\bs b(x,s)$ is vector valued and $\bs I$ the $d \times d$ unit
matrix. 
As discussed further above, each kernel $K^n$ governs a CTRW
process, which is subdiffusive with coefficient $\beta$,
meaning that waiting times have the power-law distribution
\begin{align*}
\pr(D^n_{k+1} - D^n_k > w) 
= 1 \wedge \frac{w^{-\beta}}{n \Gamma(1-\beta)}.
\end{align*}
Jumps are biased according to the external force $\bs b(x,t)$, which
is evaluated at the time of a jump.
It can be checked that the four limit statements from
Theorem~\ref{thm:spctim-conv} are satisfied with 
$b_i(x,s)$ as given, $a_{ij}(x,s) = \sigma^2 \delta_{ij}$ (Kronecker-delta), $c(x,s) = 0$ and 
$\Pi(dz, dw | x,s) = \delta_0(dz) w^{-1-\beta}\,dw / \Gamma(1-\beta)$ (Here $\delta_0$ denotes the Dirac measure concentrated at $0 \in \Rd$).

The continuum process $\{(A(r),D(r))\}_{r \ge 0}$ is then 
such that $D(r)$ is a $\beta$-stable subordinator,
i.e.\ a L\'evy process with non-decreasing sample paths
\cite{Bertoin04}. 
Since $\Pi$ puts infinite measure on the positive real line,
$D(r)$ is strictly increasing, 
and $A(r)$ is a diffusion process with constant diffusivity 
$\sigma^2 \cdot \bs I$ and drift given by $\bs b(A(r), D(r))\,dr$. 
Its representation as a stochastic differential equation is
\begin{align*}
d A(r) = b(A(r-),D(r-))\,dr + \sigma^2\,dW(r)
\end{align*}
where $W(r)$ is $d$-dimensional standard Brownian motion.
\end{example}

\section{The Semi-Markov property}

We have seen that from the sequence $(A^n_k, D^n_k)_{k \in \mathbb N_0}$
the trajectory of a CTRW $X^n(t)$ can be uniquely reconstructed. 
The $\Rd$-valued CTRW $X^n(t)$ is not a Markov process, but the 
$\spctim$-valued process $(X^n(t), V^n(t))$ is;
Here, $V^n(t)$ is the ``residence time'' of a CTRW (i.e.\ the time which
has passed since its last jump), defined as
\begin{align*}
V^n(t) = t - D^n_k, \quad \text{ where $k$ is such that } \quad D^n_k \le t < D^n_{k+1}.
\end{align*}
To see the Markov property, note that for any $\tau > 0$, 
\begin{align*}
  \ex[f(X(t+\tau), V(t + \tau)) | (X_s, V_s): s \le t]
= \ex[f(X(t+\tau), V(t + \tau)) | (X_s, V_s): s \in [D^n_k, t]]
\\
= \ex[f(X(t+\tau), V(t + \tau)) | (X_s, V_s): s \in [t - V(t), t]]
= \ex[f(X(t+\tau), V(t + \tau)) | (X_t, V_t)],
\end{align*}
where the first equality follows from the renewal property of the CTRW,
and the last equality from $X(s) = X(t)$ and $V(s) = V(t) + s - t$
on $s \in [t-V(t), t]$. 

The following theorem shows that if the convergence
$$\left\lbrace\left(A^n_{L(nr)}, D^n_{L(nr)}\right)\right\rbrace_{r \ge 0} \stackrel{J_1}{\to}\{(A(r),D(r))\}_{r \ge 0}$$
of the space-time valued processes holds as in 
Theorem~\ref{thm:spctim-conv}, then 
the CTRWs \& residence time processes $\{(X^n(t), V^n(t))\}_{t \ge 0}$
also converge.

\begin{theorem}
\label{th:SM-conv}
Let $K^n$ be a sequence of transition kernels on $\spctim$, 
$(x_0,s_0)$ a starting point, $X^n(t)$ the corresponding sequence of
CTRWs, and $V^n(t)$ the sequence of residence time processes. 
If assumptions 1.\ and 2.\ of Theorem~\ref{thm:spctim-conv} hold
and if the process $D(r)$ has a.s.\ strictly increasing sample paths, 
then the process sequence $\{(X^n(t),V^n(t))\}_{t \ge s_0}$ converges
weakly (with respect to the Skorokhod $J_1$ topology). 
The limiting $\spctim$ valued process $\{(X(t),V(t))\}_{t \ge s_0}$
has sample paths which are right-continuous with existing left-hand
limits, and is given by 
\begin{align*}
X(t) &= \lim_{\eps \downarrow 0} \xi(t + \eps), 
&
\xi(t) &:= \lim_{\eps \downarrow 0} A(E(t) - \eps)
\\
V(t) &= \lim_{\eps \downarrow 0} \eta(t + \eps), 
&
\eta(t) &:= t - \lim_{\eps \downarrow 0} D(E(t) - \eps)
\end{align*}
where $\{(A(r),D(r))\}_{r \ge 0}$ is as in 
Theorem~\ref{thm:spctim-conv} and 
\begin{align*}
E(t) := \inf\{r \ge 0: D(r) > t\}.
\end{align*}
\end{theorem}

A proof is given in the appendix. 

\begin{remark}
The limiting process from Theorem~\ref{th:SM-conv} has the intuitive
shorthand representation 
\begin{align*}
X(t) &= (A_- \circ E)_+(t), & V(t) &= t - (D_- \circ E)_+(t), 
& t &\ge s_0
\end{align*}
where $\circ$ is the composition of trajectories and a minus / plus 
sign in the subscript
denotes the left-continuous / right-continuous version of a trajectory.
\end{remark}

The special case where $D(r)$ is a strictly increasing L\'evy process
(i.i.d.\ increments) the process $E(t)$ has continuous sample paths
and is called the inverse subordinator (see e.g.\ 
\cite{invSubord}).
In the often discussed model of subdiffusion with space- and
time-dependent forcing \cite{StrakaHenry}, $A(r)$ is a diffusion
process, with drift evaluated at the times $D(r)$ \cite{Weron2008}. The time-change 
of $A(r)$ by $r = E(t)$ is called subordination.
Theorem~\ref{th:SM-conv} above, however, holds in the general situation, 
where jumps of a walker may be coupled with (i.e.\ are not independent
of) the waiting times \emph{in the limit as} $n \to \infty$.
In Example~\ref{ex:AD}, there is a dependence of the jumps on the 
preceding waiting time $w$, through $\bs b(x,s+w)$.
If the external force $\bs b(x,s)$ is evaluated at the beginning
of a waiting time, another type of dependence arises, which results
in different sample paths \cite{Angstmann2015};
In the limit $n\to \infty$, however, this dependence vanishes.
Jumps and waiting times remain coupled in the limit if and only if
$\Pi(B|x,s) > 0$, 
where
$B:= \{(z,w) \in \spctim : z \neq 0, w > 0\}$
(but this is not the case in Example~\ref{ex:AD}).
The case where $\pi(B|x,s) = 0$ for such $B$ for all $(x,s)$ is called the
\emph{uncoupled} case. 

\begin{remark} \label{rm:uncoupled}
In the uncoupled case, the CTRW limit has the simpler 
representation
\begin{align*}
X(t) &= A(E(t)), \quad t \ge s_0
\end{align*}
(see \cite{StrakaHenry}).
\end{remark}

\begin{example}
The sequence $X^n(t)$ of subdiffusive CTRWs from Example~\ref{ex:AD}
thus converges to the process $X(t) = (A_- \circ E)_+(t)$, and according 
to Remark~\ref{rm:uncoupled} $X(t) = A(E(t))$. 
The probability densities of $X(t)$, if they exist, solve the fractional
Fokker-Planck equation
\begin{align*}
\frac{\del}{\del t} p(x,t) = \mathcal L \frac{\del^{1-\beta}}{\del t^{1-\beta}} p(x,t), \quad p(x,0) = p_0(x),
\end{align*}
where the Fokker-Planck operator is given by
\begin{align*}
\mathcal L f(x,t) = 
\sigma^2 \Delta_{xx} f(x,t) 
- \nabla_x [\bs b(x,t) f(x,t)],
\end{align*}
(compare \cite{HLS10PRL,BaeumerStraka16,Magdziarz2014}).
Note that since $0$ is the start of a waiting time for all particles,
the initial condition assumes that all particles
have age $0$, i.e.\ that $\pr(V_0 = 0) = 1$.
\end{example}

\section{Discrete Semi-Markov Processes}

Theorems \ref{thm:spctim-conv} and \ref{th:SM-conv} provide limit
theorems which are applicable to a large class of CTRW limits, 
and show that the Markov property holds for CTRWs as well as for their
limit processes.
In this section, we assume that a CTRW limit process $X(t)$ is given, 
and construct a sequence of discrete CTRWs $X^n(t)$ which converges
to $X(t)$. 
Any member $X^n(t)$ assumes values on a discrete spatial lattice, 
in a fashion similar to \cite{Angstmann2015}.
Rather than integrating into the history of $X^n(t)$, however, our goal
here is to implement the Markovian dynamics of $(X^n(t),V^n(t))$,
and thus to become able to directly incorporate distributions of
residence times into the initial condition.
With discrete Markovian dynamics, the master equations for the evolution 
of probability functions of $(X^n(t),V^n(t)$ can then be 
straightforwardly implemented, see the next section.

\paragraph{Simplifying Assumptions.}
Recall that according to Theorem~\ref{thm:spctim-conv}, a CTRW limit
process is characterized by the coefficient functions 
$a_{ij}(x,s)$, $b_i(x,t)$, $c(x,t)$ and the space-time L\'evy kernel 
$\Pi(dz,dw|x,t)$.
For simplicity, we narrow down the class of CTRW limits that we
consider in the remainder of this article.
We assume only nearest neighbor jumps on the spatial lattice, 
which entails that the L\'evy measures have the representation
$\Pi(dz,dw|x,t) = \delta_0(dz) \psi(dw|x,t)$ for some measures $\psi$ on
$(0,\infty)$, and the dynamics are \emph{uncoupled}. 
We further assume that $\psi(dw|x,t) = \psi(dw)$, i.e.\
waiting times are homogeneous.
To avoid speaking of degenerate CTRW limits, we assume that the measure
$\psi(dw)$ is infinite (i.e.\ has a non-integrable singularity at $0$, 
even though $\psi(\{0\}) = 0$)\footnote{Indeed, if $\psi(dw)$ is a finite measure, then the process $D(r)$
is a step process, and hence the limiting CTRW is again a CTRW.}.
Finally, we focus on the one-dimensional case and assume that 
$c(x,s) \equiv c$ and $a(x,s) \equiv a$ are constant.

\paragraph{}
Instead of working with the measure $\psi(dw)$, it is more convenient
in our setting to analyse the (right-continuous) tail function 
$\Psi(w) := \psi((w,\infty))$ instead.
Then infiniteness of $\psi(dw)$ translates to 
$\lim_{w \downarrow 0} \Psi(w) = \infty$, and the L\'evy measure
property to $\int_0^1 \Psi(w)\,dw < \infty$, as can be seen by
integration by parts. 
The typical example to
have in mind is $\Psi(w) = w^{-\beta} / \Gamma(1-\beta)$ for 
$\beta \in (0,1)$, and
$\psi(dw) = \beta w^{-1-\beta}/ \Gamma(1-\beta)\,dw$ (subdiffusion).
To arrive at a computational algorithm for the master equations
for the laws of $(X^n_t, V^n_t)$, we need to give a sequence of 
transition kernels $K^n(dz,dw|x,s)$ which are supported on a lattice, 
and which satisfy \eqref{eq:K1}--\eqref{eq:K4}.
With this in mind, we define 
\begin{align}\label{eq:Hn} 
H^n(w) := \frac{\Psi(\left\lceil w/\tau \right\rceil \tau - \tau_2)}{n}, 
\quad w \ge 0,
\end{align}
where we define the ceiling function as 
$\lceil x \rceil := \min\{k \in \mathbb Z: k > x\}$.
The constants $\tau$ and $\tau_2$ depend on $n$ and are defined as follows:
\begin{align*}
\tau_1(n) := \Psi^{-1}(n), \quad 
\tau_2(n) :=  c/n, \quad \tau(n) := \tau_1(n) + \tau_2(n), \quad c \ge 0.
\end{align*}
It can then be checked that $H^n(0) = 1$, that $H^n$ is right-continuous and
decreases to $0$ as $w \to \infty$.
Thus $H^n(w)$ is the tail function of a probability measure supported on the lattice
$\tau \mathbb N = \{\tau, 2\tau, 3\tau, \ldots \}$. 
Since $H^n(w)$ is of finite variation,
one can define the Lebesgue-Stieltjes measure $dH^n$ via
$dH^n((a,b]) = H^n(b) - H^n(a)$.
Note however, that since $H^n(w)$ is decreasing, this measure is 
negative.
Since $H^n(w)$ is piecewise constant, with jumps in the set $\tau \mathbb N$,
$-dH^n$ is a discrete probability measure on $\tau\mathbb N$.
Finally, define a sequence of CTRW processes $X^n(t)$ 
(and their residence time processes $V^n(t)$) via their transition kernel: 
\begin{multline} \label{eq:discrete-K}
K^n(dz,dw|x,s) := -dH^n(w) \left[ \ell(x,s+w) \delta_{-\chi}(dz) 
+ r(x,s+w) \delta_\chi(dz) \right]
\\
\ell(x,s) := (1-\chi b(x,s)/a)/2, 
\quad r(x,s) := (1+\chi b(x,s)/a)/2, 
\quad \chi^2 = a/n.
\end{multline}
The probabilities $r(x,s)$ and $\ell(x,s)$ to jump right/left need of course
to be positive, which is satisfied for small enough $\chi$.
Given a starting point $x_0$ on the lattice 
$\chi \mathbb Z = \{k \chi: k \in \mathbb Z\}$, the CTRW $X^n(t)$ will
remain on this lattice at all times. 
Moreover, if the starting time is chosen from the lattice $\tau \mathbb N$,
then all jump times will also lie on this lattice. 

The following Lemma will show that $K^n(dz,dw|x,s)$ satisfies requirements
\eqref{eq:K2} and \eqref{eq:K4}:

\begin{lemma}
\label{lem:Psi}
Let $\Psi(w)$ and $H^n(w)$ be as above. 
Then the following two equalities hold:
\begin{align*}
&\lim_{\eps \downarrow 0} \lim_{n \to \infty} n \int_{(0,\eps]} w\,dH^n(w) = -c,
\\
&\lim_{n \to \infty} n \int_{(0,\infty)} g(w)\,dH^n(w) = \int_{(0,\infty)} g(w) \, d\Psi(w),
\end{align*}
where $g$ ranges over all real-valued differentiable functions with
compact support in $(0,\infty)$. 
\end{lemma}

A proof is given in the appendix. 
The following result may be interpreted as the consistency of our discrete
Semi-Markov scheme:

\begin{theorem}
\label{th:consistency}
Let the simplifying assumptions as set out above hold, and consider the
sequence of discrete CTRWs $X^n(t)$ with residence time processes $V^n(t)$, 
for $n \in \mathbb N$, defined via the kernels \eqref{eq:discrete-K}
and starting point $x_0$ at time $0$.
Then $(X^n(t), V^n(t))$ converges\footnote{You guessed it! 
Weakly with respect to Skorokhod's $J_1$ topology.} 
to the process $(X(t), V(t))$ as given in Theorem~\ref{th:SM-conv}.
That is, $X(t) = A(E(t))$, where 
\begin{enumerate}
\item[i)]
$A(r)$ is a diffusion process with constant 
diffusivity $a$ and drift $b(A(r),D(r))$, with $A(0) = x_0$
\item[ii)]
$D(r)$ is an independent subordinator (strictly increasing L\'evy process) with 
drift $c$ and L\'evy measure $\psi(dw)$, and
\item[iii)]
$E(t) = \inf\{ u: D(u) > t \}$ is the inverse subordinator.
\end{enumerate}
The process $V(t) = t - (D_- \circ E)_+(t)$ tracks the residence time of $X(t)$,
and $(X(t),V(t))$ satisfy the Markov property.
\end{theorem}

\begin{proof}
Noting that $\chi \downarrow 0$ as $n \to \infty$
and $-dH^n(w) \to \delta(dw)$ (weakly), it is straightforward
to show that \eqref{eq:K1} and \eqref{eq:K3} are satisfied by 
\eqref{eq:discrete-K}. 
Due to Lemma~\ref{lem:Psi}, \eqref{eq:K2} and \eqref{eq:K4} hold as
well. 
Since the L\'evy measure is infinite, $D(r)$ is strictly increasing
a.s., 
and thus Theorem~\ref{th:SM-conv} applies.
\end{proof}

For large $n$, we may hence assume that the distribution of 
$(X^n(t),V^n(t))$ for $t \in \tau\mathbb N$ will be a good 
approximation for the distribution of the
of the CTRW limit $(X(t),V(t))$. 
In the next section we compute these distributions.

\section{Algorithm}

We can now derive a time-stepping algorithm which calculates the
probability functions of the discrete process $(X^n(t),V^n(t))$,
whose state space is $\chi \mathbb Z \times \tau \mathbb N$, and
whose time-steps lie in $\tau \mathbb N$. 
Recall that for $k \in \mathbb N$,  $H^n(k \tau)$ denotes the
probability that a waiting time of the CTRW $X^n(t)$ is $(k+1)\tau$
or longer.
Therefore conditioning on $V^n(0) = v \in \tau \mathbb N$, we are conditioning on the waiting
time being longer  than $v$, that is $v + \tau$ or longer.
Hence observing the transition kernel \eqref{eq:discrete-K} we find:
\begin{align} \label{eq:prob-step}
\begin{split}
&\pr(X^n(\tau) \in dy, V^n(\tau) \in du | X^n(0) = x, V^n(0) = v)
\\
&= \frac{H^n(v+\tau)}{H^n(v)} \delta_x(dy) \delta_{v+\tau}(du)
+ \left(1 - \frac{H^n(v+\tau)}{H^n(v)}\right)
[\ell(x,\tau) \delta_{x-\chi}(dy) + r(x,\tau) \delta_{x+\chi}(dy)]
\delta_0(du)
\end{split}
\end{align}
where $x \in \chi \mathbb Z$, $v \in \tau \mathbb N$
and where we set $H^n(0) := 1$. 
Writing
$$\xi(i,j,k) := \pr(X_{k\tau} = i\chi, V_{k\tau} = j\tau), \quad
h(j) = H^n(j \tau),$$
we may then write a master equation for the evolution of these probabilities: 
The first term on the right-hand side of \eqref{eq:prob-step}
corresponds to the case where a particle remains on its site $x$
for another time step $\tau$, and hence we have
\begin{align*}
&\xi(i,j,k+1) = \frac{h(j)}{h(j-1)} \xi(i,j-1,k), \quad j \ge 1.
\end{align*}
The second term corresponds to the complementary case: 
a particle jumps to one of the neighboring lattice sites
$x-\chi$ or $x + \chi$, and its age is reset to $0$.
At a given lattice site $i$ the probability mass is hence obtained
by a weighted sum over all residence times of the neighbouring lattice
sites:
\begin{align*}
&\xi(i,0,k+1) 
\\
&= \sum_{j=0}^\infty \left( 1 - \frac{h(j+1)}{h(j)} \right)
\left(\ell((i+1)\chi,(k+1)\tau)\xi(i+1,j,k)
+r((i-1)\chi,(k+1)\tau)\xi(i-1,j,k)\right).
\end{align*}
The CTRW limit density $\rho(x,t)$ of the process $X(t)$ can then be approximated through
\begin{align*}
&\rho(x,t) \approx \sum_{j=0}^{\infty} \xi(i,j,k), 
\quad i = [x/\chi], \quad k = [t/\tau].
\end{align*}

In practice, the algorithm runs on a finite grid
$$\{-L, -L+1, \ldots, 0, \ldots, L-1, L\} \times \{0, 1, \ldots, R\},$$
representing the state space,
and one has to impose additional boundary conditions. 

\paragraph{Spatial Boundary conditions.}
We only consider the one-dimensional case. 
For \emph{absorbing}, or Dirichlet boundary conditions $\rho(b) = 0$
where $b$ is a boundary point, a walker is removed if it walks off
the lattice. 
That is, we set $\ell(-m\chi, k\tau) = 0$ and $r(m\chi, k\tau) = 0$
for all $0 \leq k \leq N$, $\ell(-m\chi, k\tau) = 0$; 
note that on the boundary site, $\ell$ and $r$ hence no longer add
to $1$.

For \emph{reflecting}, or Neumann boundary conditions, a particle
remains at a boundary site whenever it would jump off the lattice, 
and adjust \eqref{eq:discrete-K} accordingly.

\paragraph{Residence time boundary conditions.}
When the residence time of a particle approaches the lattice end at
$R$, we could force it to jump to a neighboring lattice site 
and reset its age to $0$.
This effectively corresponds to a tail function
$\Psi(w) \mathbf 1\{\Psi(w) \ge \Psi(R \tau)\}$, 
and hence we term this the \emph{cutoff} boundary condition.

Below, however, we assume that upon reaching the end of the lattice
at $R$, a particle is not forced to jump, and allow it to remain
at its site $x$ with residence time $R$ if it would not otherwise
jump. 
That is, we set 
$$\xi(i,R,k+1) = \frac{h(R)}{h(R-1)} \xi(i,R-1,k)
+ \frac{h(R+1)}{h(R)} \xi(i,R,k).$$
This means that particles with residence time $R\tau$ remain unchanged
for a geometrically distributed number of time steps, with 
parameter $1 - h(R+1)/h(R)$.
In the scaling limit, this corresponds to an exponential distribution, 
whose rate is
$$\gamma(\mathbf R) := \frac{\psi(\mathbf R)}{\Psi(\mathbf R)}, 
\quad \mathbf R := \lim_{n \to \infty} \tau R$$
(note that as $n \to \infty$, we have $\tau \downarrow 0$ and 
$R \to \infty$). 
This effectively corresponds to a tail function
$$\Psi(w)\mathbf 1\{\Psi(w) \ge \Psi(\mathbf R)\}
+ \Psi(\mathbf R) e^{-\gamma(\mathbf R) (w-\mathbf R)} \mathbf 1\{\Psi(w) < \Psi(\mathbf R)\}$$
and hence we term this procedure the \emph{cross-over} boundary 
condition.

\paragraph{}
Assume now as a general initial condition a probability measure 
$\mu(dx,dv)$, 
and that the aim is to calculate 
$$\int \pr(X_t \in dy, V_t \in du | X_0 = x, V_0 = v) \mu(dx,dv),$$
To this end, we set $\xi(i,j,0) = \mu([i\chi, (i+1)\chi) \times [j\tau, (j+1)\tau))$,
and simply run our algorithm with this initial condition.
Note that $R$ needs to be chosen large enough in order to 
avoid cut-off or cross-over effects for $\Psi(w)$ as discussed above. 
A safe choice is always
$$R = \max\{j: \xi(i,j,0) > 0, |i| \le L\} + N,$$
where $N$ denotes the number of time steps, 
though it may of course be infeasible in cases where 
$\mu$ has unbounded support.

\section{Examples}

Within our framework, we may now compute (approximations of) 
probability distributions of CTRW limits, with varying initial 
residence times, for a variety of models. 
In particular, we may assume any subordinator $D(r)$, and thus
treat a variety of non-Markovian behaviours (see Table~\ref{tab:models}).
Two main regimes occur, depending on whether $\Psi(t)$
has integrable tails or not. 
In the former case, $V(t)$ admits the equilibrium distribution
\begin{align} \label{eq:equi-V}
\pi(B) = \frac{c}{c+g} \delta_0(B)
+ \frac{1}{c+g} \int_B \Psi(w)\,dw
\end{align}
where $g := \int_0^\infty \Psi(w)\,dw$ and $\delta_0$ denotes a Dirac
measure at $0$ \cite{Hor72}.
In the latter case, there exists an invariant measure, but it is
infinite, and hence an equilibrium cannot be reached. 

\paragraph{Tempering.}
Throughout, $\beta \in (0,1)$. 
The tail function $\Psi(w) = w^{-\beta} / \Gamma(1-\beta)$ in the
subdiffusive case is not integrable. 
The \emph{tempered} subdiffusive case is obtained by multiplying the L\'evy
density with an exponential $e^{-\gamma w}$ \cite{Rosinski07}. 
The tail function becomes 
\[ \Psi(t|\beta, \gamma) = \frac{\beta}{\Gamma(1-\beta)} \int_t^\infty w^{-1-\beta} e^{-\gamma w}\,dw 
= \frac{t^{-\beta}e^{-\gamma t} - \gamma \Gamma(1-\beta, t)}{{\Gamma(1-\beta)}}, \quad \gamma \ge 0\]
where $\Gamma(\beta,t)$ denotes the upper incomplete Gamma function. 
This modification makes $\Psi(t|\beta,\gamma)$ integrable for 
$\gamma > 0$, that is, $g < \infty$. 
CTRW limits with these ``tempered dynamics'' appear subdiffusive on
short time scales and diffusive on longer time scales \cite{Stanislavsky2008,StrakaThesis,Gajda2010}. 
Note that for $\gamma = 0$ the above reduces to the subdiffusive case.

\paragraph{Subordinator with drift.}
If the subordinator $D(r)$ has a positive drift constant $c > 0$, 
the resulting growth of $D(r)$ at very short times is proportional
to $cr$. 
Accordingly, the inverse subordinator $E(t)$ also grows
linearly, proportionally to $t/c$ for short times\footnote{A law of the 
iterated logarithm applies for the precise limit, see \cite{Bertoin04}.}.
For larger time scales, the jumps of $D(r)$ will dominate the drift
$c$, if $g \gg c$ (or if $g = \infty$ in the case where $\Psi(w)$
is not integrable).
This means that for long times, the temporal evolution appears 
subdiffusive if $\gamma = 0$ \cite{StrakaThesis}. 
The case $\gamma > 0$ and $c > 0$
has been 
examined in \cite{StrakaThesis}: 
$E(t) \sim t/c$ grows linearly for small time scales. 
For long time scales, $E(t)$ also grows linearly, although with a 
smaller slope. 
To our knowledge, the cross-over between the two regimes at intermediate 
time scales has not been 
looked at in detail but we predict it will show the signatures of
subdiffusive behavior. 

Finally, in the case where $c > 0$ and $\gamma > 0$, 
by the above $g < \infty$, and $c$ and $g$ admit
a nice physical interpretation: 
At equilibrium, $c/(c+g)$ is the fraction of ``mobile'' particles
which have residence time $0$, and $g/(c+g)$ is the fraction 
of ``immobile'' particles, which have been trapped for a time $w$
distributed as $\Psi(w)\,dw / (c+g)$. 
We deem this to be an interesting tempered extension of the so called
``fractal mobile/immobile model'' of \cite{SchumerMIM}.
If $\gamma = 0$, there exists no equilibrium, and all 
mobile particles eventually seep into the immobile phase.

\begin{table}
\centering
\begin{tabular}{r|c|c}
model & tempering parameter $\gamma$ & temporal drift $c$
\\
\hline
\hline
Subdiffusion & $\gamma = 0$ & $c = 0$
\\
\hline
tempered subdiffusion & $\gamma > 0$
& $c = 0$
\\
\hline
fractal mobile-immobile
& $\gamma = 0$
& $c > 0$
\\
\hline
tempered fractal mobile-immobile
& $\gamma > 0$
& $c > 0$
\end{tabular}
\caption{\label{tab:models}
We consider four cases of non-Markovian temporal evolutions, governed
by the inverse subordinator $E(t)$.}
\end{table}

\paragraph{Varying the initial residence time.}
When modelling subdiffusion or tempered subdiffusion, the standard
assumption is that the first waiting time starts at $t =0$, 
which translates to the initial condition $\mu(dx,dv) = \rho_0(dx)
\delta_0(dv)$ (all particles have residence time $0$, and their location 
is distributed according to $\rho_0(dx)$, typically $\rho_0(dx) = 
\delta_0(dx)$, \cite{HLS10PRL}). 
Subdiffusive CTRWs are known to exhibit ageing, which is an indefinite
slowing down of the dynamics as $t$ increases. 
\cite{BarkaiCheng} consider dynamics of CTRW limits where the system
has been prepared at a time $-t_a$, and study the dynamics on the
interval $(0,t)$, for which e.g.\ a Fokker-Planck equation has been
derived in \cite{Busani}. 
This relates to our approach by taking as initial condition the 
probability distribution
$\mu(dx, dv) = \pr[X(t_a) \in dx, V(t_a) \in dv | X(0) = 0, V(0) = 0]$, 
and calculating the probability distributions of
$\pr[X(s)\in dx, V(s) \in dx|\mu]$ for $s \in (0,t)$.


\begin{figure}[h]
\centering
\includegraphics[scale=1]{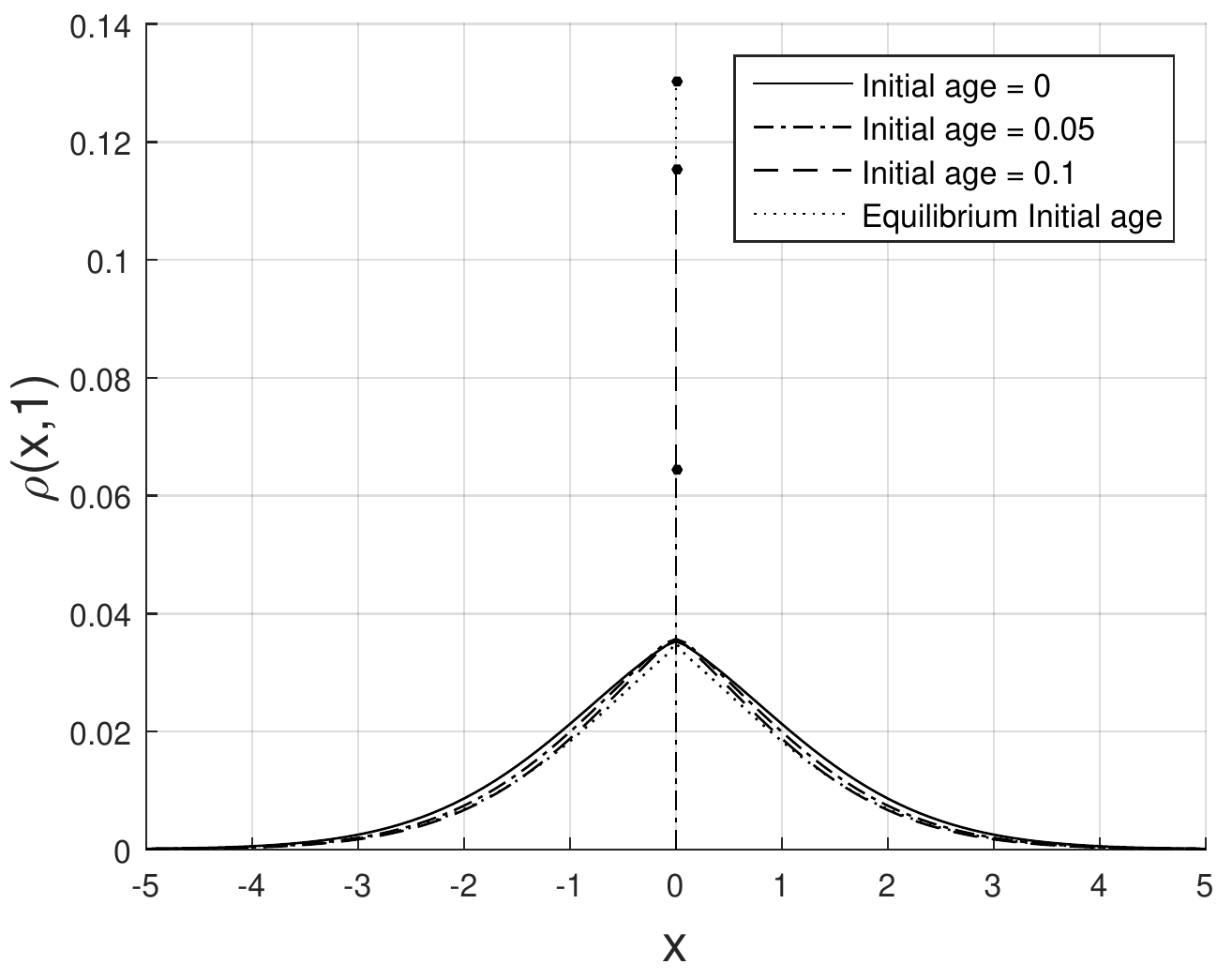}
\caption{\label{fig:rho2} 
Probability distribution $\rho(x,t)$ of a subdiffusive process at 
$t = 1$ with varying initial age condition.
The ``equilibrium initial age'' condition is as in \eqref{eq:equi-V},
with a tempering parameter $\gamma = 1$.
For positive initial residence time, 
this law has a point mass at $x=0$. 
The remaining mass admits a continuous distribution.}
\end{figure}

In the subdiffusive setting, we now examine the impact of a varying
initial residence time on the probability function of a CTRW limit. 
In particular, we calculate the ``Green's functions'' 
$\pr[X(t) \in dx, V(t) \in dv | X(0) = 0, V(0) = v]$
where $v \ge 0$. 
For simplicity, we assume symmetric nearest neighbor
jumps with reflecting boundary condition, and a fractional parameter
$\beta = 0.9$.
Figure~\ref{fig:rho2}, with $v = 0$ shows the distinctive cusp shape 
of the probability density of
subdiffusive CTRW limits (see e.g.\ \cite{Metzler2000}).
On the other hand, if conditioning on $X(0) = 0$ and $V(0) = v$ where $v$ is positive, 
the particle is trapped at $0$, and stays there until time $t$
with probability $\Psi(v+t) / \Psi(v)$; 
Compare \cite[Th~4.1]{SemiMarkovCTRW} which provides a formula for the 
conditional distribution 
$\pr[X(t) \in dx, V(t) \in dv | X(t) = 0, V(t) = v]$. 
Hence the joint distribution of $(X(t), V(t))$, conditioned on 
$X(0) = 0, V(0) = v$, has an atom of mass $\Psi(v+t) / \Psi(v)$
at $(0,v+t)$.
This atom reflects in the marginal distribution of $X(t)$, as 
shown in Figures~\ref{fig:rho2} and \ref{fig:res}. 
The remaining probability mass, as given in \cite[Th~4.1]{SemiMarkovCTRW}, 
is absolutely continuous. 
As $v \to \infty$, the weight 
$\Psi(v+t) / \Psi(v)$ of this atom increases towards $1$.

\begin{figure}
\centering
\includegraphics[scale=0.9]{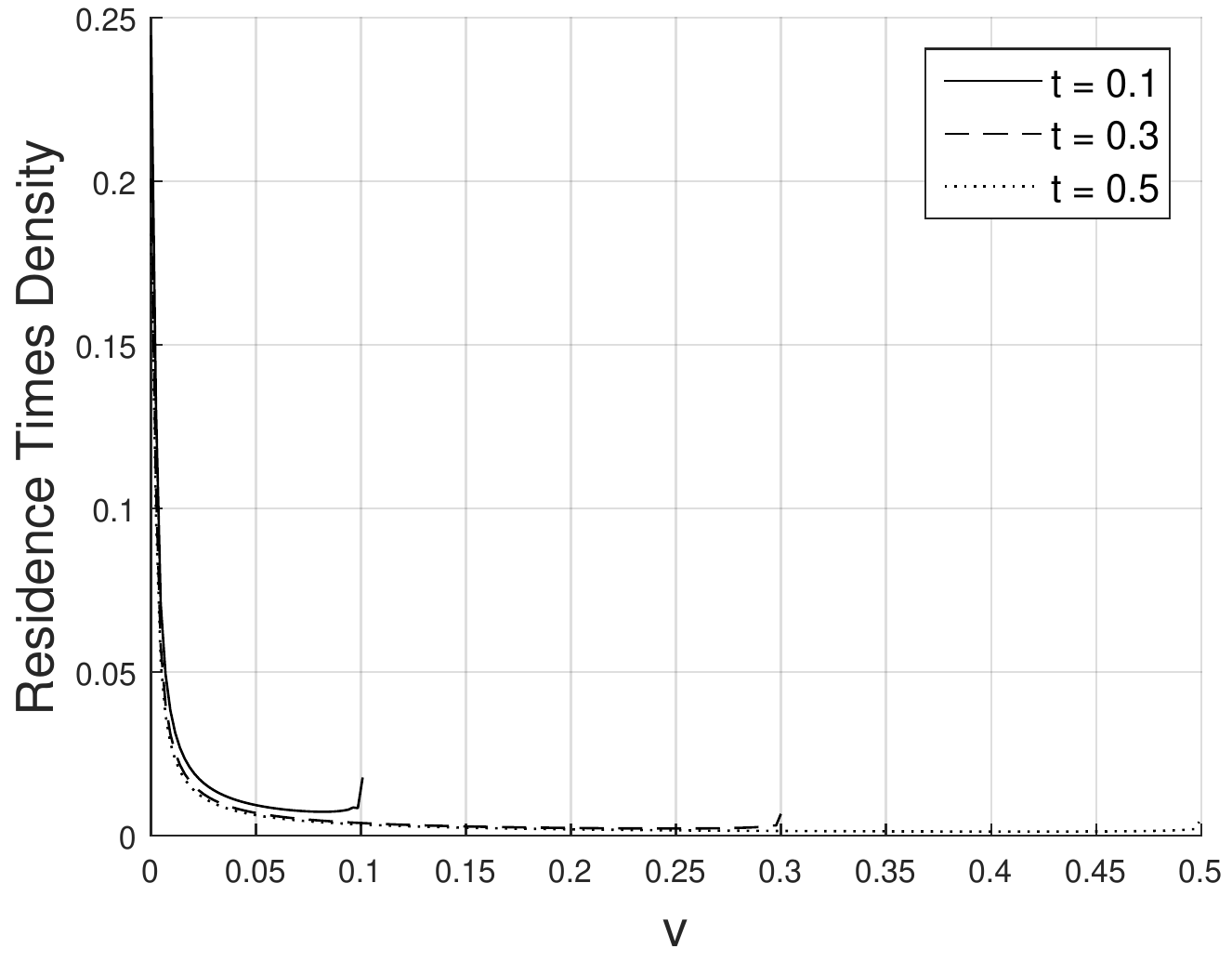}\\
\includegraphics[scale=0.9]{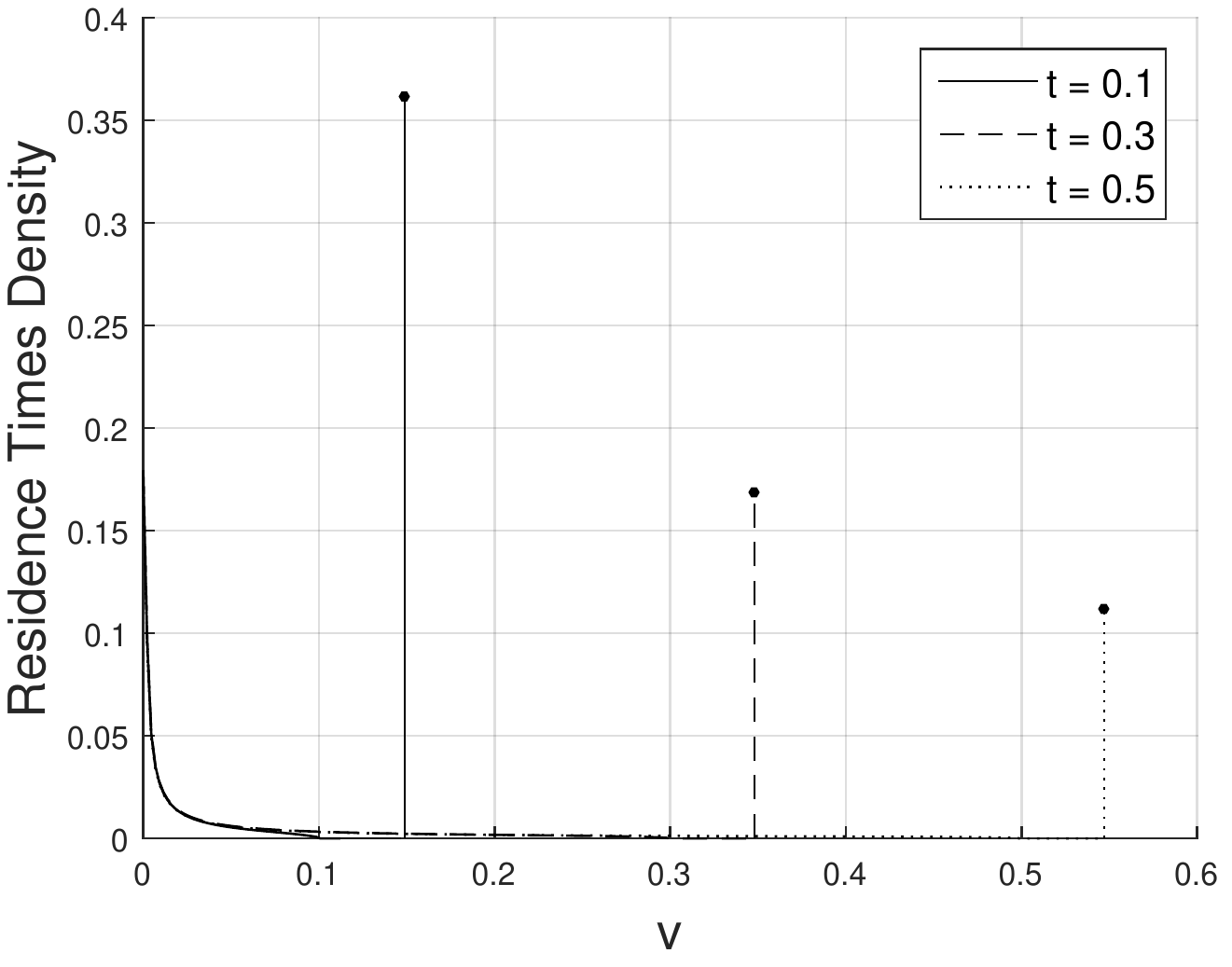}
\caption{\label{fig:res} 
The probability distribution of $V(t)$ given $V(0) = v$ for $v=0$ (left) 
and $v = 0.05$ (right), for the tempered subdiffusive case. 
For $v = 0$, the densities resemble the arcsine distribution. 
For $v > 0$, there is a point mass $\Psi(v+t)/\Psi(t)$ at $v+t$, 
with the remaining probability mass continuously distributed on the
remaining interval $(0,t)$.
Parameters are $\gamma = 1$, $n = 25$, $a = 1$, $c = 0$ and $\beta = 0.9$.}
\end{figure}

\paragraph{Evolution of the residence time distribution.}
Figure \ref{fig:res} describes the evolution of the densities of the
residence time process $V_t$. 
Again we consider the subdiffusive case with $\beta = 0.9$
as in the previous paragraph. 
If $V(0) = v = 0$, we have due to self-similarity
$V(t) / t \stackrel{d}{=} V(1)$, 
and the distribution of $V(1)$ follows the arcsine law
$$\pr[V(1) \in ds] = \frac{\sin \beta \pi}{\pi} s^{\beta-1} (1-s)^{-\beta}\,ds,$$
compare \cite[Prop~3.1]{Bertoin04}.
If  $V(0) = v > 0$, the distribution of $V(t)$ has an atom
at $v+t$, whose weight increases to $1$
as $v \to \infty$, compare the discussion in the previous paragraph.
In the tempered case, for $t \to \infty$ the distribution approaches
\eqref{eq:equi-V}.


\paragraph{Computational accuracy.}

\begin{figure}
\includegraphics[width=\textwidth]{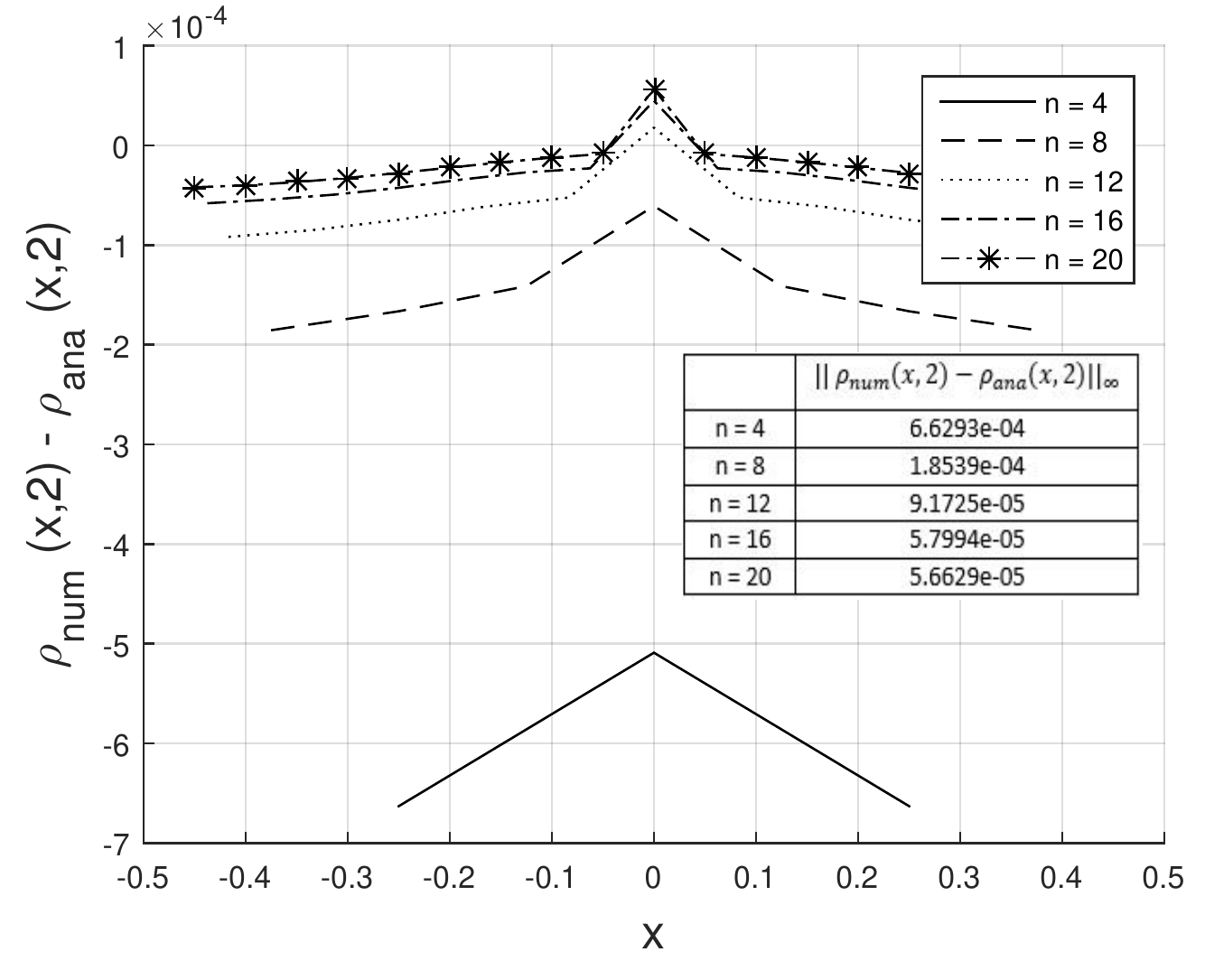}
\caption{\label{fig:errors} 
Computational errors for the standard subdiffusion equation with
$\beta = 0.8$, using our algorithm from Section 5.
As $n$ increases, results become seemingly more accurate.}
\end{figure}

Exact analytical solutions to the symmetric subdiffusion equation
are available, and we check our computed densities against these
solutions. 
Following \cite{Angstmann2015}, use the series representation
$$\rho(x,t) = 1 + \sum_{k=1}^\infty (-1)^k 2 \exp(-(2k\pi)^2t) \cos(2k\pi x)$$
for the solution $\rho(x,t)$ to the ``standard'' fractional diffusion equation
$$\frac{\del \rho(x,t)}{\del t} = a \frac{\del^{1-\beta}}{\del t^{1-\beta}} \frac{\del^2 \rho(x,t)}{\del x^2}.$$
As shown e.g.\ in \cite{HLS10PRL}, the corresponding CTRW limit process
is given by symmetric nearest neighbor jumps, $c = \gamma = 0$ and
$\Psi(w) = w^{-\beta}/\Gamma(1-\beta)$. 
Figure~\ref{fig:errors} displays the computational errors for this case,
which seem to stabilize as the densities of the discrete CTRW approach
the CTRW limit (as $n \to \infty$).

%
%
%

\section{Conclusion}

Similar in spirit to \cite{Angstmann2015}, we have derived an algorithm
for the computation of probability distributions of CTRW limits,
which is based on the stochastic process rather than the 
fractional Fokker-Planck equation. 
Additionally, our approach calculates the residence time, or age
of a walker, which is of independent physical interest, 
and which may be of use for the modelling of non-Markovian diffusion
with distributed age initial condition.

In \cite{Angstmann2016} it is shown that the discrete stochastic
processes approach from \cite{Angstmann2015} is also applicable
to model reaction-diffusion problems and nonlinear interactions. 
Our approach above assumes that particles do not interact, and 
there are severe technical obstacles in extending the above 
mathematical rigour to CTRW limit processes which interact via 
reactions, chemotaxis, or otherwise. 
It is straightforward, however, to write down master equations
with interactions using the Semi-Markov formalism, and thus to
calculate mass distributions, see \cite{StrakaFedotov14}.
The work here may be viewed as an extension to \cite{StrakaFedotov14}
which can model general subordinated particle dynamics. 

In order to focus on the main ideas, we have only considered CTRWs
with nearest neighbor jumps and homogeneous waiting times.
By varying the coefficients $a$, $b$, $c$ and $K$ and by possibly
making them vary in space and time, one can arrive at a variety
of different models; for three such models, see
\cite{BaeumerStraka16}.
It is possible to generalize the Semi-Markov algorithm from 
Section 5 to coupled and non-local jump operators, 
given the formulas derived in \cite{SemiMarkovCTRW}, 
though this may of course require much larger computational effort.

\subsection*{Acknowledgements}
P.~Straka was supported by the UNSW Science Early Career Research Grant and
the Australian Research Council's Discovery Early Career Research Award.

\appendix
\section{Proofs}

\begin{proof}[Proof of Theorem \ref{thm:spctim-conv}]
We  apply Th~IX.4.8 in \cite{JacodShiryaev}.
The process
$\{(A^n_{L(nr)},D^n_{L(nr)})\}_{r \ge 0}$ is a semimartingale in
$\spctim$, in the sense of the cited book. 
Relative to the truncation function
\begin{align*}
h(y,w) &=
 \begin{cases}
  (y,w) &\text{ if } \|y\| < 1 \text{ and } 0 < w < 1 \\
 (0,0) &\text{ else }
 \end{cases}
\end{align*}
its characteristics are
$((\mathbf B^n,\mathbf C^n), \mathbf A^n,  \Pi^n)$, where
 \begin{align*}
\mathbf  B^n_i(t) &= \int_0^t b^n_i(A^n_{L(nr)}, D^n_{L(nr)})dr,
 & b^n_i(x,t) &=n \iint  h_i(y,w)~ K^n(dy,dw|x,t) \\
\mathbf C^n(t) &= \int_0^t \tilde c^n(A^n_{L(nr)}, D^n_{L(nr)}) dr,
& \tilde c^n(x,t) &=n \iint  h_{d+1}(y,w)~ K^n(dy,dw|x,t) \\
\mathbf A^n_{ij}(t) &= \int_0^t \tilde a^n_{ij}(A^n_{L(nr)}, D^n_{L(nr)})dr,
& \tilde a^n_{ij}(x,t) &=n \iint  (h_i h_j)(y,w)~ K^n(dy,dw|x,t)
\end{align*}
\begin{align*}
\Pi^n(dy,dw;dr) = K^n(dy,dw|A^n_{L(nr)}, D^n_{L(nr)})dr
\end{align*}
and where $(h_i h_j) (y,w) = h_i(y,w) h_j(y,w)$.
Observing that
\begin{align}
\lim_{n \to \infty} \tilde c^n(x,t) &= c(x,t) + \iint h_{d+1}(y,w) \Pi(dy,dw|x,t), \\
\lim_{n \to \infty} \tilde a^n_{ij}(x,t)
&= a_{ij}(x,t) + \iint (h_i h_j)(y,w) \Pi(dy,dw|x,t), \quad 1 \le i,j \le d
\end{align}
one verifies that the assumptions of Th~IX.4.8 in \cite{JacodShiryaev} are satisfied.
\end{proof}

\begin{proof}[Proof of Theorem~\ref{th:SM-conv}]
We apply Proposition 2.3 in \cite{StrakaHenry}, which states the
following: The mapping
\begin{align*}
(\alpha, \delta) \mapsto 
\left((\alpha_-, \delta_-) \circ \epsilon_-\right)_+,
\end{align*}
defined for c\`adl\`ag\footnote{French acronym
for right-continuous with left-hand limits} trajectories
$\alpha$ and $\delta$ in $\Rd$ resp.\ $\R$, where $\delta$ is 
increasing and unbounded, and where 
$\epsilon(t) := \inf\{r: \delta(r) > t\}$,
is continuous at all trajectories $(\alpha,\delta)$ where 
$\delta$ is strictly increasing. 
As before, $\circ$ denotes a composition of trajectories, 
and a $+/-$ in the subscript denotes the right-continuous resp.\ 
left-continuous version of a trajectory.
Continuity is with respect to the (metrizable) Skorokhod $J_1$ topology
on the set of all such trajectories \cite{JacodShiryaev}. 

Next, apply the continuous mapping theorem \cite{Billingsley1968}:
Since the processes $(A^n_{L(nr)},D^n_{L(nr)})$ converge to 
$(A(r),D(r))$ as $n \to \infty$, and $D(r)$ is strictly increasing 
(almost surely), the sequence of their images
$(X^n(t),G^n(t))$ must converge to the image $(X(t), G(t))$.
Here, $G^n(t) := (D^n_- \circ E^n_-)_+(t)$, 
$E^n(t) = \inf\{r: D^n_{L(nr)} > t\}$, and 
$G(t) = (D_- \circ E)_+ (t)$ (note that $E(t)$ has a.s.\ increasing
sample paths). 
It is tedious but not too difficult to check that 
$(X^n(t),G^n(t)$ and $(X(t),G(t))$ are really the images of
$(A^n_{L(nr)},D^n_{L(nr)})$ and $(A(r),D(r))$ for the above mapping.

Finally, in a similar fashion mapping the process 
$G^n(t)$ to the process $V^n(t) = t - G^n(t)$ also defines a
continuous mapping, hence $V^n(t)$ also converges to $V(t)$. 
\end{proof}

\begin{proof}[Proof of Lemma \ref{lem:Psi}]
The measure $dH^n$ is concentrated at the steps
$\tau \mathbb N = \{\tau, 2\tau, 3\tau, \ldots\}$
of the function $H^n$.
We use this and Lebesgue-Stieltjes integration by parts
\cite{Hewitt1960} to calculate
\begin{multline*}
n \int_{(0,\eps]} w\,dH^n(w) 
= n \int_{[\tau, \eps]} w\,dH^n(w)
= [nw H^n(w)]^{\eps}_\tau - n\int_{[\tau,\eps]} H^n(w)\,dw
\\
= \eps n H^n(\eps) - \tau n H^n(\tau)
- \int_{[\tau, \eps]} \Psi\left( \lceil w/\tau \rceil
\tau - c/n \right)\,dw.
\end{multline*}
Now examine these terms individually as $n \to \infty$:
\begin{align*}
&\eps nH^n(\eps) \to \eps \Psi(\eps)
\\
&\tau n H^n(\tau) = \tau (\Psi(\tau - c/n))
\sim (\tau_1 + \tau_2) \Psi(\tau_1)
= \tau_1 n + c
\\
&\Psi\left( \lceil w/\tau \rceil
\tau - c/n \right) \to \Psi(w)
\end{align*}
where $\sim$ means the two sequences have the same limit. 
By dominated convergence,
the integral of the third expression converges to 
$\int_{(0,\eps]} \Psi(w)\,dw$.
Since $\Psi$ is integrable at $0$, 
$\Psi(w) \le C w^{-\beta}$ at $w \downarrow 0$ where
$\beta \in (0,1)$. 
Hence $\tau_1 n = \tau_1 \Psi(\tau_1) \to 0$.
Now letting $\eps \downarrow 0$ gives the first statement. 

The second statement follows by integration by parts and dominated
convergence:
\begin{align*}
-n \int_{(0,\infty)} g(w)\,dH^n(w)
= n \int_{(0,\infty)}g'(w) H^n(w)\,dw
\\
\to \int_{(0,\infty)} g'(w) \Psi(w)\,dw
= -\int_{(0,\infty)} g(w) \, d\Psi(w)
\end{align*}
(note that the boundary terms vanish by definition of $g$).
\end{proof}

\end{document}